\documentclass{article}
\usepackage{amssymb}
\usepackage{amsfonts}
\usepackage[latin1]{inputenc}
\usepackage{amsmath}
\usepackage[dvips]{graphicx}
\usepackage{color}

\newtheorem{theorem}{Theorem}

\newtheorem{definition}[theorem]{Definition}

\newtheorem{lemma}[theorem]{Lemma}

\newtheorem{proposition}[theorem]{Proposition}
\newtheorem{remark}[theorem]{Remark}

\newenvironment{proof}[1][Proof]{\noindent\textbf{#1.} }{\ \rule{0.5em}{0.5em}}
\begin{document}

\title{Matricial representation of period doubling cascade}
\author{{\normalsize Luc\'{\i}a Cerrada}$^{a,b}${\normalsize , Jes\'{u}s San
Mart\'{\i}n}$^{a,c}${\normalsize ,} \\
%EndAName
{\normalsize a) Departamento de Matem\'{a}tica Aplicada, E.U.I.T.I. }\\
{\normalsize \ UPM. 28012-\ Madrid, SPAIN.}\\
{\normalsize \ b) Departamento de Matem\'{a}tica Aplicada y computaci\'{o}n,
}\\
{\normalsize \ ETS de Ingenier\'{\i}a ICAI.}\\
{\normalsize \ Universidad Pontificia Comillas de Madrid. }\\
{\normalsize \ 28015-Madrid, SPAIN.}\\
{\normalsize \ c) Departamento de F\'{\i}sica Matem\'{a}tica y de Fluidos,}\\
{\normalsize \ Facultad de Ciencias. Universidad Nacional de Educaci\'{o}n a
Distancia.}\\
{\normalsize \ 28040-Madrid, SPAIN.}\\
{\normalsize \ Corresponding author: Jes\'{u}s San Mart\'{\i}n. e-mail:
jsm@dfmf.uned.es}}
\maketitle

\begin{abstract}
Starting from the cycle permutation $\sigma _{2^{k}}$ associated with the $%
2^{k}$-periodic orbit of the period doubling cascade we obtain the inverse
permutation $\sigma _{2^{k}}^{-1}$. Then we build a matrix permutation
related to $\sigma _{2^{k}}^{-1}$, which includes the visiting order of the $%
2^{k}$-periodic orbit points.

After some manipulations a recurrence relation of matricial representation
of period doubling cascade is obtained. Finally the explicit matricial
representation is reached.
\end{abstract}

\section{Introduction}

The period doubling cascade \cite{FEIGENBAUM78, FEIGENBAUM79} is a very
well-known mechanism of transition to chaos and it is commonly observed in
science. It can be found in chemistry \cite{ISTVAN06}, biology \cite{INNOCENTI07},
physics \cite{LETELLIER07}, to mention a few fields. This can
give an idea of both its extraordinary importance and the mathematical foundations
that supports it. Any advance in the understanding or the reformulation of
period doubling cascade will have immediate effects in many fields as
those mentioned above.

In a period doubling cascade successive orbits of ever increasing (doubling) period are created
as a parameter control is varied. The process that underlies in this
phenomenon is well-known as well as the scaling among successive bifurcation
parameters \cite{FEIGENBAUM78, FEIGENBAUM79}. We are going to pay attention
to the problem not from analytical but topological point of view, where
several questions remain still open.

If the $2^{k}$-periodic orbit of the cascade is considered, then it can be
characterized by its symbolic sequence \cite{PAVUY66, METROPOLIS73}. That
is, a sequence like $CI_{1}I_{2\cdots }I_{2^{k}-1}$ where $I_{i\text{ }}$ can
be $R$ or $L$, depending on whether the points of the orbit are located to the right or
left of the critical point $C$ of the function undergoing the period
doubling cascade. A question immediately arises. How can we distinguish one $R
$ $\left( L\right) $ from the others? Which are their relative positions?.
This question is answered by the celebrated kneading theory \cite{MILNOR88}.
By using this theory every point of the orbit is associated with a number,
then you order the numbers one by one and you obtain the relative positions.
Obviously if the orbit has many points, let us say $2^{5000}$, it is better
to use some kind of sieve to sort the points. That is what is made in \cite%
{SMARTIN09}. In that paper it was found the permutation that orders the points
of the $2^{k}$ periodic orbit, $k$ arbitrary, in the period-doubling cascade. To get this goal, the $2^{k}$
points of the orbit are labeled by their natural position in the straight line
as follows $\left\{ C_{\left( 1,2^{k}\right) },\text{ }C_{\left(
2,2^{k}\right) },C_{\left( 2,2^{k}\right) },\cdots C_{\left(
2^{k},2^{k}\right) }\right\} $. Then the permutation $\ {\large \sigma }%
_{2^{k}}=\left( \sigma _{\left( 1,2^{k}\right) },\sigma _{\left(
2,2^{k}\right) },\cdots ,\sigma _{\left( 2^{k},2^{k}\right) }\right) $,
indicates $\sigma _{\left( i,2^{k}\right) }$, the number of iterates from $C$ needed to
get the point $C_{\left( i,\text{ }2^{k}\right) }$, that is, $f^{\sigma
_{\left( i,2^{k}\right) }}\left( C\right) =C_{\left( i,2^{k}\right) }$. The
inverse permutation ${\Large \sigma }_{2^{K}}^{-1}$ gives the order the
different points of the orbit are visited (see \cite{SMARTIN09} for more details).

Nonetheless some unsolved problems remain.

\begin{enumerate}
\item[i)] How to make $k\rightarrow \infty $ in the $2^{k}$-periodic orbit? Such
happens in a period doubling cascade. If we had a matricial representation
we would be able to do so because we have mathematical tools in algebra to
get it. Furthermore, in this way we would have a mathematical tool to study
analytically the transition to chaos, that is, when $k\rightarrow \infty$.

\item[ii)] How to deal with a dynamical system undergoing a period doubling cascade
that has $2^{m}$ points of $2^{k}$ period orbit $\left( m\leq k\right) $
occupied by particles. That is, the physical system has $2^{k}$ states
generated in a period doubling cascade, and $2^{m}$ of those states are
occupied by particles. Obviously, we can calculate the evolution of any
particle by using the permutation $\sigma _{2^{k}}$ for all of them, but in
this way we would have to use the permutation $2^{m}$ times, what is
completely useless if $2^{m}$ is very big. Matricial representation
is again the solution, given that we can represent $2^{k}$ states of the
system by a single column vector with $2^{k}$ rows, the rows are zero
except when they are occupied by particles. So the evolution of the system
is given by the matricial representation of $\sigma _{2^{k}}$ acting on the
column vector. On the other way, the particles occupying the $2^{m}$ points
can have physical properties: spin, charge, etc. since every particle evolution is
distinguishable. We shall discuss this topic later.

\end{enumerate}

The idea of using matrices in dynamical systems is not new, in fact, it is
naturally associated with them. Every periodic orbit is associated
univocally with a permutation and this permutation with its permutation
matrix. Matrices have a pivotal role in dynamical systems, for example, they
are used either to describe the minimal periodic orbit structure with the
transition matrix \cite{HALL94} or in kneading theory \cite{MILNOR88} with
the kneading matrix.

In this paper we are going to get a matricial representation of period
doubling cascade orbits, by using the permutation $\sigma _{2^{k}}$ given in
\cite{SMARTIN09}. In order to achieve this goal the paper is organized as
follows. Firstly, some propositions and lemmas are proven to get $\sigma
_{2^{k}}^{-1}$ from $\sigma _{2^{k}}$, given that the sought for matricial
representation is associated with $\sigma _{2^{k}}^{-1}$ instead of $\sigma
_{2^{k}}$. Then, after a few manipulations we arrive to a new permutation related
to $\sigma _{2^{k}}^{-1}$, the so-called disordered permutation. From this
new permutation we obtain the permutation matrix associated with it. Then we
get the matricial representation by a recurrence relation. Finally we obtain
an explicit matricial representation of period doubling cascade. After
theoretical results we will expose conclusions where new fundamental
questions will arise.

\section{Definitions and notation}

Along this paper, we will use basic concepts introduced in \cite{SMARTIN09}.
For the sake of self containing notation and clarity we repeat the following
definitions.

\begin{definition}\label{DEF_1}
The set $\left\{ C_{\left( 1,q\right) }^{\ast },C_{\left( 2,q\right) }^{\ast
}\ldots ,C_{\left( q,q\right) }^{\ast }\right\} $ will denote the ascending
\linebreak $\left( \text{descending}\right) $ cardinality order of the orbit
$\ O=\left\{ C,f\left( C\right) ,\cdots ,f^{q-1}\left( C\right) \right\} $
when the unimodal map $f$ has a minimum (maximum) in $C$.
The point $C_{\left( i,q\right) }^{\ast }$ is
defined as the $i$-th cardinal of the q-periodic orbit.
\end{definition}

\begin{definition}\label{DEF_2}
The natural number $\sigma _{\left( i,q\right) }$, $i=1,\cdots ,q$ will
denote the number of iterations of $f$ such that $f^{\sigma _{\left(
i,q\right) }}\left( C\right) =C_{\left( i,q\right) }^{\ast }$, $i=1,\cdots,q$.
\end{definition}

\begin{figure}[tbh]
\begin{center}
\includegraphics[width=0.65\textwidth]{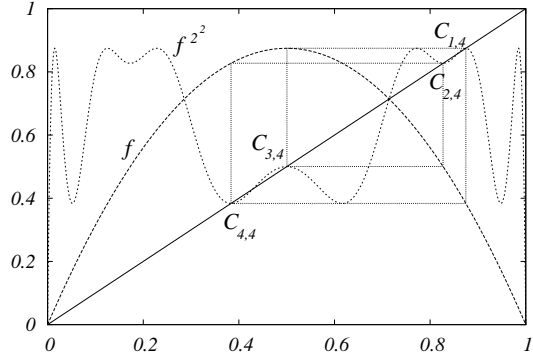}
\end{center}
\caption{\label{FIG_1}A $4$-periodic orbit. $f(C)=C^{\ast}_{1,4}$, $f^2 (C)=C^{\ast}_{4,4}$, $f^3 (C)=C^{\ast}_{2,4}$, $f^4 (C)=C^{\ast}_{3,4} = C$,
 with $\sigma_{1,4}=1$, $\sigma_{2,4}=3$, $\sigma_{3,4}=4$, $\sigma_{4,4}=2$. }
\end{figure}

\begin{remark}
The geometrical meaning of definitions \ref{DEF_1} and \ref{DEF_2} is depicted in figure \ref{FIG_1}. Introducing the usual order of $\mathbb{R}$ in the set of the real numbers $O=\{ C=f^4 (C), f(C), f^2 (C), f^3 (C) \}$ it results $ f^2 (C)< f^4 (C) < f^3 (C) < f(C) $, then we label these real numbers according to the usual (ascending or descending) order of $\mathbb{N}$, as $C^{\ast}_{4,4} < C^{\ast}_{3,4} < C^{\ast}_{2,4} < C^{\ast}_{1,4}$ (the label runs in the index $"i"$ of $C^{\ast}_{i,4}$). In this way, we can forget the real numbers given by $f^{\sigma(i,4)}$ and speak of $C^{\ast}_{i,4}$. That is the meaning of the cardinality order in a set whose cardinal is $q=4$. In fact, what is important is to know how many iterates are needed to move from $C^{\ast}_{i,q}$ to $C^{\ast}_{j,q}$, that is, to get from label $"i"$ to label $"j"$.
\end{remark}

\begin{definition}
We denote as $\sigma _{q}$ the permutation $\sigma _{q}=\left( \sigma
_{\left( 1,q\right) },\sigma _{\left( 2,q\right) ,\cdots ,}\sigma _{\left(
q,q\right) }\right) ,$ that is the $q$-tuple formed by the $\sigma _{\left(
i,q\right) }$.
\end{definition}

In particular, along this paper we will take $q=2^{k}$.

\begin{definition}\label{DEF_4}
Let $\left\{ a_{1},a_{2,\cdots ,}a_{n}\right\} $ be a sequence of real
numbers. We define the Reflection of the sequence $\left\{ a_{1},a_{2,\cdots
,}a_{n}\right\} $ with increment $\alpha$, denoted as $R\left(
a_{1},a_{2,\cdots ,}a_{n};\alpha \right) $, as the sequence of real numbers
given by
\begin{equation*}
R\left( a_{1},a_{2,\cdots ,}a_{n};\alpha \right) =\left\{ a_{1},a_{2,\cdots
,}a_{n},a_{n}+\alpha ,a_{n-1}+\alpha ,\cdots ,a_{1}+\alpha \right\}
\end{equation*}
\end{definition}

We shall introduce the rest of definitions when the need for them arises.

\section{Theorems}

Without loss of generality we shall work with unimodal maps having a
maximun at the critical point $C$.

The next theorem is proven in \cite{SMARTIN09}, we reformulate it, by using
definition \ref{DEF_4}, for our convenience as follows.

\begin{theorem}
\label{T_1}Let $\ f:I\subset
%TCIMACRO{\U{211d} }%
%BeginExpansion
\mathbb{R}
%EndExpansion
\rightarrow I$ be an unimodal map, depending on a parameter, undergoing a
period doubling cascade. Let $\left\{ C_{\left( 1,2^{k}\right) }^{\ast
},C_{\left( 2,2^{k}\right) }^{\ast }\ldots ,C_{\left( 2^{k},2^{k}\right)
}^{\ast }\right\} $ be the cardinality ordering of the $2^{k}-$periodic
superstable orbit of the cascade, given by $f^{\sigma _{\left(
i,2^{k}\right) }}\left( C\right) =C_{\left( i,2^{k}\right) }^{\ast }$, %
$i=1,\cdots ,2^{k}.$ Then the elements of the permutation%
\begin{equation*}
{\large \sigma }_{2^{k}}=\left( \sigma _{\left( 1,2^{k}\right) },\sigma
_{\left( 2,2^{k}\right) },\cdots ,\sigma _{\left( 2^{k},2^{k}\right)
}\right) \equiv \left( \sigma _{2^{k}}\left( 1\right) ,\sigma _{2^{k}}\left(
2\right) ,\cdots ,\sigma _{2^{k}}\left( 2^{k}\right) \right)
\end{equation*}%
satisfy the recurrence relation%
\begin{equation}
\begin{array}{cc}
\sigma _{2^{k+1}}\left( n\right) =2\sigma _{2^{k}}\left( n\right) -1 &  \\
& n=1,2,3,\cdots ,2^{k} \\
\sigma _{2^{k+1}}\left( 2^{k}+n\right) =2\sigma _{2^{k}}\left(
2^{k}-n+1\right) &
\end{array}
\label{F_00}
\end{equation}
\end{theorem}

\bigskip

The reformulated theorem \ref{T_1} allows us to obtain the inverse
permutation of ${\large \sigma }_{2^{k}}$ by a recurrence relation.

\begin{proposition}
\label{P_1}The permutation, denoted by $\sigma _{2^{k}}^{-1}$, defined by the recurrence relation%
\begin{equation}
\begin{array}{c}
{\large \sigma }_{2^{k}}^{-1}\left( 2n-1\right) ={\large \sigma }%
_{2^{k-1}}^{-1}\left( n\right)  \\
\\
{\large \sigma }_{2^{k}}^{-1}\left( 2n\right) =2^{k}+1-{\large \sigma }%
_{2^{k-1}}^{-1}\left( n\right)
\end{array}%
1\leq n\leq 2^{k-1}\ \ n,k\in
%TCIMACRO{\U{2115} }%
%BeginExpansion
\mathbb{N}
%EndExpansion
\label{F_0}
\end{equation}%
with seed value

\begin{equation*}
\sigma _{2^{0}}^{-1}\left( 1\right) =1
\end{equation*}%
is the inverse permutation of $\sigma _{2^{k}}$.
\end{proposition}

\begin{proof}

We will prove the proposition by induction on $k$.

\textbf{i)} ${\large \sigma }_{2^{1}}^{-1}\circ {\large \sigma }%
_{2^{1}}\left( n\right) =n$, $n=1,2$. The proof is straightforward.

\textbf{ii)} We assume by hypothesis of induction that ${\large \sigma }%
_{2^{k}}^{-1}$ is the inverse permutation of ${\large \sigma }_{2^{k}}$ and
we will prove that ${\large \sigma }_{2^{k+1}}^{-1}$ is the inverse of $%
{\large \sigma }_{2^{k+1}},$ that is,%
\begin{equation*}
\sigma _{2^{k+1}}^{-1}\circ \sigma _{2^{k+1}}\left( n\right) =\sigma
_{2^{k+1}}^{-1}\left( \sigma _{2^{k+1}}\left( n\right) \right) =n,\ \
n=1,\cdots ,2^{k},2^{k}+1,\cdots ,2^{k+1}
\end{equation*}%
We will distinguish two cases.

\textbf{a)} $n=1,\cdots ,2^{k}$

By theorem \ref{T_1} it yields%
\begin{equation*}
\sigma _{2^{k+1}}^{-1}\left( \sigma _{2^{k+1}}\left( n\right) \right)
=\sigma _{2^{k+1}}^{-1}\left( 2\sigma _{2^{k}}\left( n\right) -1\right)
\end{equation*}

By using $\left( \ref{F_0}\right) $ it results

\begin{equation*}
\sigma _{2^{k+1}}^{-1}\left( 2\sigma _{2^{k}}\left( n\right) -1\right)
=\sigma _{2^{k}}^{-1}\left( \sigma _{2^{k}}\left( n\right) \right) =n
\end{equation*}%
where the last step holds by the induction hypothesis.

\textbf{b)} $n=2^{k}+1,\cdots ,2^{k+1}$

We write $n=2^{k}+m$ with $m=1,\cdots ,2^{k}$, so
\begin{equation*}
\sigma _{2^{k+1}}^{-1}\left( \sigma _{2^{k+1}}\left( n\right) \right)
=\sigma _{2^{k+1}}^{-1}\left( \sigma _{2^{k+1}}\left( 2^{k}+m\right) \right)
\end{equation*}

By Theorem $\ref{T_1}$ it results%
\begin{equation*}
\sigma _{2^{k+1}}^{-1}\left( \sigma _{2^{k+1}}\left( 2^{k}+m\right) \right)
=\sigma _{2^{k+1}}^{-1}\left( 2\sigma _{2^{k}}\left( 2^{k}-m+1\right) \right)
\end{equation*}

Then by using $\left( \ref{F_0}\right) $ it yields%
\begin{eqnarray*}
\sigma _{2^{k+1}}^{-1}\left( 2\sigma _{2^{k}}\left( 2^{k}-m+1\right) \right)
&=&2^{k+1}+1-\sigma _{2^{k}}^{-1}\left( \sigma _{2^{k}}\left(
2^{k}-m+1\right) \right) = \\
&=&2^{k+1}+1-\left( 2^{k}-m+1\right) =2^{k}+m=n
\end{eqnarray*}%
where the hypothesis induction has been used.
\end{proof}

\begin{lemma}
\label{L_1}The permutations ${\large \sigma }_{2^{k}}^{-1}$ and ${\Large %
\sigma }_{2^{k-2}}^{-1}$ satisfy the recurrence relation%
\begin{equation*}
\sigma _{2^{k}}^{-1}\left( j\right) =\left\{
\begin{array}{c}
\sigma _{2^{k}}^{-1}\left( 4m\right) =2^{k-1}+\sigma _{2^{k-2}}^{-1}\left(
m\right) \\
\\
\sigma _{2^{k}}^{-1}\left( 4m-1\right) =2^{k-1}+1-\sigma
_{2^{k-2}}^{-1}\left( m\right) \\
\\
\sigma _{2^{k}}^{-1}\left( 4m-2\right) =2^{k}+1-\sigma _{2^{k-2}}^{-1}\left(
m\right) \\
\\
\sigma _{2^{k}}^{-1}\left( 4m-3\right) =\sigma _{2^{k-2}}^{-1}\left( m\right)%
\end{array}%
\begin{array}{c}
\text{with }1\leq m\leq 2^{k-2}\text{, }1\leq j\leq 2^{k} \\
k\geq 2%
\end{array}%
\right.
\end{equation*}
\end{lemma}

\begin{proof}

\textbf{i)} According to (\ref{F_0}), $\sigma _{2^{k}}^{-1}\left( 2n-1\right) =\sigma
_{2^{k-1}}^{-1}\left( n\right) $ \ $1\leq n\leq 2^{k-1}$

By using (\ref{F_0}) again we get%
\begin{equation*}
\sigma _{2^{k}}^{-1}\left( 2n-1\right) =\sigma _{2^{k-1}}^{-1}\left(
n\right) =\left\{
\begin{array}{c}
\sigma _{2^{k-1}}^{-1}\left( 2m-1\right) =\sigma _{2^{k-2}}^{-1}\left(
m\right) \text{ \ \ if\ \ \ }n=2m-1 \\
\\
\sigma _{2^{k-1}}^{-1}\left( 2m\right) =2^{k-1}+1-\sigma
_{2^{k-2}}^{-1}\left( m\right) \text{ \ \ if\ \ \ }n=2m%
\end{array}%
k\geq 2\right.
\end{equation*}%
that yields, after substituting $n=2m-1$ and $n=2m$,%
\begin{equation*}
\begin{array}{cc}
\sigma _{2^{k}}^{-1}\left( 4m-3\right) =\sigma _{2^{k-2}}^{-1}\left( m\right)
&  \\
& 1\leq m\leq 2^{k-2},\text{ \ }k\geq 2 \\
\sigma _{2^{k}}^{-1}\left( 4m-1\right) =2^{k-1}+1-\sigma
_{2^{k-2}}^{-1}\left( m\right) &
\end{array}%
\end{equation*}

\textbf{ii)} According to (\ref{F_0}), $\sigma _{2^{k}}^{-1}\left( 2n\right) =2^{k}+1-\sigma
_{2^{k-1}}^{-1}\left( n\right) $ \ $1\leq n\leq 2^{k-1}$

We proceed as before to obtain%
\begin{equation*}
\sigma _{2^{k}}^{-1}\left( 2n\right) =\left\{
\begin{array}{c}
2^{k}+1-\sigma _{2^{k-1}}^{-1}\left( 2m-1\right) =2^{k}+1-\sigma
_{2^{k-2}}^{-1}\left( m\right) \text{ \ \ if \ \ }n=2m-1 \\
\\
2^{k}+1-\sigma _{2^{k-1}}^{-1}\left( 2m\right) =2^{k}-2^{k-1}+\sigma
_{2^{k-2}}^{-1}\left( m\right) \text{ \ \ if \ \ }n=2m%
\end{array}%
k\geq 2\right.
\end{equation*}%
and therefore%
\begin{equation*}
\begin{array}{cc}
\sigma _{2^{k}}^{-1}\left( 4m-2\right) =2^{k}+1-\sigma _{2^{k-2}}^{-1}\left(
m\right) &  \\
& 1\leq m\leq 2^{k-2},\text{ \ }k\geq 2 \\
\sigma _{2^{k}}^{-1}\left( 4m\right) =2^{k-1}+\sigma _{2^{k-2}}^{-1}\left(
m\right) &
\end{array}%
\end{equation*}%
which proves the lemma.
\end{proof}

\begin{remark}
\label{R_1}Given the $2^{k}-$periodic orbit, the inverse permutation%
\begin{equation}
\sigma _{2^{k}}^{-1}=%
\begin{pmatrix}
1 & 2 & \cdots  & n & \cdots  & 2^{k} \\
\sigma _{2^{k}}^{-1}\left( 1\right)  & \sigma _{2^{k}}^{-1}\left( 2\right)
&  & \sigma _{2^{k}}^{-1}\left( n\right)  &  & \sigma _{2^{k}}^{-1}\left(
2^{k}\right)
\end{pmatrix}
\label{F_18}
\end{equation}%
indicates that we reach the position $\sigma _{2^{k}}^{-1}\left( n\right) $ after
$n$ $\left( n=1,\cdots ,2^{k}\right) $ iterates from $C$ (see
\cite{SMARTIN09}). So, the visiting order of the cardinals of the orbit is given by the sequence
\begin{equation}
\sigma _{2^{k}}^{-1}\left( 1\right) \rightarrow \sigma _{2^{k}}^{-1}\left(
2\right) \rightarrow \cdots \rightarrow \sigma _{2^{k}}^{-1}\left(
2^{k}\right)  \label{F_19}
\end{equation}
\end{remark}

\bigskip That idea is reflected in the permutation defined as%
\begin{equation*}
\begin{pmatrix}
\sigma _{2^{k}}^{-1}\left( n\right)  \\
\sigma _{2^{k}}^{-1}\left( n+1\right)
\end{pmatrix}%
\text{ with }n=1,\cdots ,2^{k}
\end{equation*}

Because of this, we introduce the following definition.

\begin{definition}
\label{D_1}Let $\sigma _{n}=%
\begin{pmatrix}
1 & 2 & \cdots & n \\
\sigma _{n}\left( 1\right) & \sigma _{n}\left( 2\right) & \cdots & \sigma
_{n}\left( n\right)%
\end{pmatrix}%
$ be a permutation, we define the disordered permutation of $\sigma _{n}$,
denoted as $\Omega \left( \sigma _{n}\right) $, as the permutation%
\begin{equation*}
\Omega \left( \sigma _{n}\right) \equiv
\begin{pmatrix}
\sigma _{n}\left( 1\right) & \sigma _{n}\left( 2\right) & \cdots & \sigma
_{n}\left( n\right) \\
\sigma _{n}\left( 2\right) & \sigma _{n}\left( 3\right) & \cdots & \sigma
_{n}\left( 1\right)%
\end{pmatrix}%
\end{equation*}%
or using compact notation%
\begin{equation*}
\Omega \left( \sigma _{n}\right) =%
\begin{pmatrix}
\sigma _{n}\left( r\right) \\
\sigma _{n}\left( r+1\right)%
\end{pmatrix}%
\text{ with }r=1,\cdots ,n
\end{equation*}%
It is assumed that $\sigma _{n}\left( n+1\right) =\sigma
_{n}\left( 1\right)$.
\end{definition}

In particular the disordered permutation associated with ${\large \sigma }_{2^{k}}^{-1}$ is%
\begin{equation}
\Omega \left( \sigma _{2^{k}}^{-1}\right) \equiv
\begin{pmatrix}
\sigma _{2^{k}}^{-1}\left( n\right)  \\
\sigma _{2^{k}}^{-1}\left( n+1\right)
\end{pmatrix}%
\text{ with }n=1,\cdots ,2^{k}  \label{F_10}
\end{equation}%
It is assumed that $\sigma _{2^{k}}^{-1}\left( 2^{k}+1\right)
=\sigma _{2^{k}}^{-1}\left( 1\right)$.

This permutation gives what point of the $2^{k}$ - periodic orbit is
reached after one iteration from a fixed point, that is, if we are located
at the point with position $\sigma _{2^{k}}^{-1}\left( n\right) $ after one
iteration we will reach the point with position $\sigma _{2^{k}}^{-1}\left(
n+1\right) $. In other words, we move from $C_{\left( \sigma
_{2^{k}}^{-1}\left( n\right) ,2^{k}\right) }^{\ast }$ to\textbf{\ }$%
C_{\left( \sigma _{2^{k}}^{-1}\left( n+1\right) ,2^{k}\right) }^{\ast }$

As the inverse permutation of $\Omega \left( \sigma
_{n}\right)$ is
\begin{equation*}
\Omega ^{-1}\left( \sigma _{n}\right) =%
\begin{pmatrix}
\sigma _{n}\left( r+1\right)  \\
\sigma _{n}\left( r\right)
\end{pmatrix}%
\text{ with }r=1,\cdots ,n
\end{equation*}%
it results
\begin{equation}
\Omega ^{-1}\left( \sigma _{2^{k}}^{-1}\right) =%
\begin{pmatrix}
\sigma _{2^{k}}^{-1}\left( n+1\right)  \\
\sigma _{2^{k}}^{-1}\left( n\right)
\end{pmatrix}%
\text{ with }n=1,\cdots ,2^{k}  \label{F_20}
\end{equation}

\section{Matricial representation of period doubling cascade}

\subsection{Matrix associated with the disordered permutation $\Omega^{-1}\left( \protect\sigma _{2^{k}}^{-1}\right)$}

Every permutation is univocally associated with its permutation matrix. So,
one could think that matricial representation of period doubling cascade is
given by the permutation matrix of $\sigma _{2^{k}}^{-1}$. However the
permutation $\sigma _{2^{k}}^{-1}$ does not give the visiting order of the $%
2^{k}-$ periodic orbit but $\Omega \left( \sigma _{2^{k}}^{-1}\right) $ $%
\left( \text{see remark \ref{R_1}}\right) $; that is, why we have built the
permutation $\Omega \left( \sigma _{2^{k}}^{-1}\right) $.\ If $P_{\Omega
\left( \sigma _{n}^{-1}\right) }$\ is the permutation matrix associated with
$\Omega \left( \sigma _{2^{k}}^{-1}\right) $ it results

\begin{equation}
\begin{pmatrix}
0 & \cdots  & 0 & 1_{i} & 0 & \cdots  & 0%
\end{pmatrix}%
P_{\Omega \left( \sigma _{n}^{-1}\right) }=%
\begin{pmatrix}
0 & \cdots  & 0 & 0 & 1_{j} & \cdots  & 0%
\end{pmatrix}
\label{F_26}
\end{equation}%
that is, a particle located at the arbitrary position\textbf{\ }$i=\sigma
_{2^{k}}^{-1}\left( n\right) ,$ $n=1,2,3,\cdots ,2^{k}$, reaches the
position $j=\sigma _{2^{k}}^{-1}\left( n+1\right) ,$ $n=1,2,3,\cdots ,2^{k}$
after one iteration.

Because we usually work with column vectors instead of row vectors, we rewrite $\left( \ref{F_26}\right)$ as

\begin{equation*}
P_{\Omega \left( \sigma _{2^{k}}^{-1}\right) }^{t}%
\begin{pmatrix}
0 \\
\vdots \\
0 \\
1_{i} \\
0 \\
\vdots \\
0%
\end{pmatrix}%
=%
\begin{pmatrix}
0 \\
\vdots \\
0 \\
0 \\
1_{j} \\
\vdots \\
0%
\end{pmatrix}%
\end{equation*}

As permutation matrices are orthogonal matrices, it results
\begin{equation*}
P_{\Omega \left( \sigma _{2^{k}}^{-1}\right) }^{t}=P_{\Omega \left( \sigma
_{2^{k}}^{-1}\right) }^{-1}=P_{\Omega ^{-1}\left( \sigma
_{2^{k}}^{-1}\right) }
\end{equation*}

That is why we shall work with $\Omega ^{-1}\left( \sigma_{2^{k}}^{-1}\right) $
instead of $\Omega \left( \sigma _{2^{k}}^{-1}\right)$.

\begin{definition}
\label{D_2}We define the matrix $P_{\Omega \left( \sigma_{n}\right) }$
as the permutation matrix associated with the disordered
permutation $\Omega \left( \sigma _{n}\right) $ , that is,%
\begin{equation*}
P_{\Omega \left( \sigma _{n}\right) }=\left( p_{ij}\right) \text{ \ \ with \
\ }p_{ij}=\left\{
\begin{array}{cl}
1 & \text{ if }i=\sigma _{n}\left( r\right) \text{ and }j=\sigma _{n}\left(
r+1\right) \text{, }r=1,2,\cdots ,n \\
0 & \text{ otherwise}%
\end{array}%
\right.
\end{equation*}%
It is assumed that $\sigma _{n}\left( n+1\right) =\sigma_{n}\left( 1\right)$.
\end{definition}

Therefore from definition $\ref{D_2}$ and $\left( \ref{F_20}\right) $ it
follows $P_{\Omega ^{-1}\left( \sigma _{2^{k}}^{-1}\right)
}=\left( p_{ij}\right) $ with
\begin{equation}
p_{ij}=\left\{
\begin{array}{cl}
1 & \text{ if }i=\sigma _{2^{k}}^{-1}\left( n+1\right) \text{ and }j=\sigma
_{2^{k}}^{-1}\left( n\right) \text{, }n=1,2,3,\cdots ,2^{k} \\
0 & \text{ otherwise}%
\end{array}%
\right.  \label{F_14}
\end{equation}%
It is assumed that $\sigma _{2^{k}}^{-1}\left( 2^{k}+1\right)=\sigma _{2^{k}}^{-1}\left( 1\right)$.

Without loss of generality $\left( \ref{F_14}\right) $ can be written as
$P_{\Omega ^{-1}\left( \sigma _{2^{k}}^{-1}\right) }=\left( p_{ij}\right) $
with

\begin{equation}
p_{ij}=\left\{
\begin{array}{cl}
1 & \text{ if }i=\sigma _{2^{k}}^{-1}\left( n\right) \text{ and }j=\sigma
_{2^{k}}^{-1}\left( n-1\right) \text{, }n=1,2,3,\cdots ,2^{k} \\
0 & \text{ otherwise }%
\end{array}%
\right.  \label{F_15}
\end{equation}%
It is assumed that $\sigma _{2^{k}}^{-1}\left( 2^{k}\right)=\sigma _{2^{k}}^{-1}\left( 0\right)$.

We want to find the explicit expression of the permutation matrix
associated with the disordered permutation $\Omega ^{-1}\left( \sigma
_{2^{k}}^{-1}\right)$. It will provide us the matricial
representation of the period doubling cascade.

\subsection{Matricial representation}

The matricial representation we are looking for, is enclosed in the
following definition.

\begin{definition}
\label{D_3}The matrices $A_{2^{k}}$, $k\in
%TCIMACRO{\U{2115} }%
%BeginExpansion
\mathbb{N}
%EndExpansion
$, are defined by%
\begin{equation*}
A_{2^{1}}=%
\begin{pmatrix}
0 & 1 \\
1 & 0%
\end{pmatrix}%
,A_{2^{2}}=%
\begin{pmatrix}
0 & 0 & \mathbf{1} & 0 \\
0 & 0 & 0 & \mathbf{1} \\
0 & \mathbf{1} & 0 & 0 \\
\mathbf{1} & 0 & 0 & 0%
\end{pmatrix}%
,A_{2^{k}}=\left(
\begin{tabular}{c|c|c}
$0$ & $A_{2^{k-2}}$ & $0$ \\ \hline
$0$ & $0$ & $I_{2^{k-2}}$ \\ \hline
$I_{2^{k-1}}^{\ast }$ & $0$ & $0$%
\end{tabular}%
\right) k\geqslant 3,
\end{equation*}%
where $I_{2^{k-2}}$ is the $2^{k-2}$ order identity matrix and $%
I_{2^{k-1}}^{\ast }=\left( a_{ij}\right) $ is the $2^{k-1}$ order matrix with $a_{ij}=\left\{
\begin{array}{cl}
1 & \text{if }i+j=2^{k-1}+1 \\
0 & \text{otherwise}%
\end{array}%
i,j=1,2,\cdots ,2^{k-1}\right. $, that is, $I_{2^{k-1}}^{\ast }=%
\begin{pmatrix}
&  &  & \mathbf{1} \\
& ^{2^{k-1)}} & \cdot &  \\
& \cdot &  &  \\
\mathbf{1} &  &  &
\end{pmatrix}%
$
\end{definition}

To prove that matrices $A_{2^{k}}$, given in the definition above, are the
looked for matricial representation we divide the proof into a sequence of
lemmas.

\begin{lemma}
\label{L_2}$P_{\Omega ^{-1}\left( \sigma _{2^{k}}^{-1}\right) }$ is a $2^{k}$
order matrix whose entries are all zero except three non-null blocks
denoted by $B_{2^{k-2}},$ $C_{2^{k-2}}$ and 
$C_{2^{k-1}}^{\ast }$ given by $B_{2^{k-2}}=\left( b_{ij}\right) $ with%
\begin{equation*}
b_{ij}=\left\{
\begin{array}{cl}
1 & \text{ if }i=\sigma _{2^{k}}^{-1}\left( 4m-3\right) \text{ \ and \ }%
j=\sigma _{2^{k}}^{-1}\left( 4m-4\right) \text{, }m=1,2,3,\cdots ,2^{k-2} \\
0 & \text{ otherwise}%
\end{array}%
\right.
\end{equation*}

$C_{2^{k-2}}=\left( c_{ij}\right) $ with
\begin{equation*}
c_{ij}=\left\{
\begin{array}{cl}
1 & \text{ if }i=\sigma _{2^{k}}^{-1}\left( 4m-1\right) \text{ \ and \ }%
j=\sigma _{2^{k}}^{-1}\left( 4m-2\right) \text{, }m=1,2,3,\cdots ,2^{k-2} \\
0 & \text{ otherwise}%
\end{array}%
\right.
\end{equation*}

$C_{2^{k-1}}^{\ast }=\left( c_{ij}^{\ast }\right) $ with %
\begin{equation*}
c_{ij}^{\ast }=\left\{
\begin{array}{cl}
1 & \text{ if }i=\sigma _{2^{k}}^{-1}\left( 2m\right) \text{ \ and \ }j=\sigma
_{2^{k}}^{-1}\left( 2m-1\right) \text{, }m=1,2,3,\cdots ,2^{k-1} \\
0 & \text{ otherwise}%
\end{array}%
\right.
\end{equation*}
\end{lemma}

\begin{proof}
Let be the set $S=\left\{ 1,2,3,\cdots ,2^{k}\right\}$. $S$ can
be divided intro three subsets%
\begin{equation}
\begin{array}{c}
S_{1}=\left\{ 4m-3\text{ }\left\vert \text{ }\right. m=1,2,3,\cdots
,2^{k-2}\right\} \\
S_{2}=\left\{ 4m-1\text{ }\left\vert \text{ }\right. m=1,2,3,\cdots
,2^{k-2}\right\} \\
S_{3}=\left\{ 2m\text{ }\left\vert \text{ }\right. m=1,2,3,\cdots
,2^{k-1}\right\}%
\end{array}
\label{F_3}
\end{equation}
such that $S=\underset{i=1}{\overset{3}{\cup }}S_{i}$, $S_{i}\cap
S_{j}=\varnothing $, $i\neq j$, $i,j=1,2,3$.

As (see (\ref{F_15})) $P_{\Omega ^{-1}\left( \sigma_{2^{k}}^{-1}\right) }=\left( p_{ij}\right)$ with
\begin{equation*}
p_{ij}=\left\{
\begin{array}{cl}
1 & \text{ if }i=\sigma _{2^{k}}^{-1}\left( n\right) \text{ \ and \ }j=\sigma
_{2^{k}}^{-1}\left( n-1\right) \text{, }n=1,2,3,\cdots ,2^{k} \\
0 & \text{ otherwise}%
\end{array}%
\right.
\end{equation*}
taking into account the partition of $S$ given by (\ref{F_3})
it follows, after substituting by $4m-3$, $4m-1$, and $2m$,
that $P_{\Omega^{-1}\left( \sigma _{2^{k}}^{-1}\right)}$ has three non - null blocks, denoted $B_{2^{k-2}}$,
$C_{2^{k-1}}^{\ast }$ and $C_{2^{k-2}}$ as outlined below.

\textbf{i) }$B_{2^{k-2}}=\left( b_{ij}\right) $ is a $2^{k-2}$ order
submatrix, with
\begin{equation}
b_{ij}=\left\{
\begin{array}{cl}
1 & \text{ if }i=\sigma _{2^{k}}^{-1}\left( 4m-3\right) \text{ \ and \ }%
j=\sigma _{2^{k}}^{-1}\left( 4m-4\right) \text{, }m=1,2,3,\cdots ,2^{k-2} \\
0 & \text{ otherwise}%
\end{array}%
\right.  \label{F_4}
\end{equation}

\textbf{ii) }$C_{2^{k-2}}=\left( c_{ij}\right) $ is a $2^{k-2}$ order
matrix, with
\begin{equation}
c_{ij}=\left\{
\begin{array}{cl}
1 & \text{ if }i=\sigma _{2^{k}}^{-1}\left( 4m-1\right) \text{ \ and \ }%
j=\sigma _{2^{k}}^{-1}\left( 4m-2\right) \text{, }m=1,2,3,\cdots ,2^{k-2} \\
0 & \text{ otherwise}%
\end{array}%
\right.  \label{F_5}
\end{equation}

\textbf{iii) }$C_{2^{k-1}}^{\ast }=\left( c_{ij}^{\ast }\right) $ is a $%
2^{k-1}$ order matrix, with
\begin{equation}
c_{ij}^{\ast }=\left\{
\begin{array}{cl}
1 & \text{ if }i=\sigma _{2^{k}}^{-1}\left( 2m\right) \text{ \ and \ }j=\sigma
_{2^{k}}^{-1}\left( 2m-1\right) \text{, }m=1,2,3,\cdots ,2^{k-1} \\
0 & \text{ otherwise}%
\end{array}%
\right.  \label{F_6}
\end{equation}
being zero the other entries in the matrix $P_{\Omega}^{-1}\left( \sigma _{2^{k}}^{-1}\right)$.
\end{proof}

The following lemma indicates where $B_{2^{k-2}},$ $C_{2^{k-2}}$
and $C_{2^{k-1}}^{\ast }$ are located in the matrix
$P_{\Omega _{\text{ }}^{-1}\left( \sigma _{2^{k}}^{-1}\right) }$.

\begin{lemma}
\label{L_3}Let $B_{2^{k-2}},$ $C_{2^{k-2}}$ and 
$C_{2^{k-1}}^{\ast }$ be the blocks introduced in lemma $\ref{L_2}$, then
\begin{equation*}
P_{\Omega ^{-1}\left( \sigma _{2^{k}}^{-1}\right) }=\left(
\begin{tabular}{c|l|c}
$0$ & $B_{2^{k-2}}$ & $0$ \\ \hline
$0$ & \multicolumn{1}{|c|}{$0$} & $C_{2^{k-2}}$ \\ \hline
$C_{2^{k-1}}^{\ast }$ & \multicolumn{1}{|c|}{$0$} & $0$%
\end{tabular}%
\right)
\end{equation*}
\end{lemma}

\begin{proof}

\textbf{i) }We want to prove that $B_{2^{k-2}}$ is the submatrix
of $P_{\Omega^{-1}\left( \sigma _{2^{k}}^{-1}\right) }$ formed
by rows $1,2,\cdots ,2^{k-2}$ and columns $2^{k-1}+1,2^{k-1}+2,\cdots
,2^{k-1}+2^{k-2}$.

According to (\ref{F_4}) the row-indexes of $B_{2^{k-2}}$ with
entry $1$ are given by $\sigma _{2^{k}}^{-1}\left( 4m-3\right) $ and by using
lemma \ref{L_1} we write
\begin{equation}
\sigma _{2^{k}}^{-1}\left( 4m-3\right) =\sigma _{2^{k-2}}^{-1}\left(
m\right) \ \text{\ with\ \ }m=1,2,3,4,\cdots ,2^{k-2},\text{ }k\geqslant 2
\label{F_7}
\end{equation}

As $\sigma _{2^{k-2}}^{-1}$ is a permutation it yields that $\sigma_{2^{k-2}}^{-1}\left( m\right) \in \left\{ 1,2,3,\cdots ,2^{k-2}\right\}$,
that is, the rows of submatrix $B_{2^{k-2}}$ with entry $1$ are located in the
$2^{k-2}$ first rows of $P_{\Omega ^{-1}\left( \sigma _{2^{k}}^{-1}\right)}$.

We now proceed in the same manner with the column-indexes that have entry
$1$ in the submatrix $B_{2^{k-2}}$.

By using (\ref{F_4}) and lemma \ref{L_1} it results that columns
with entry $1$ are given by
\begin{equation}
\sigma _{2^{k}}^{-1}\left( 4m-4\right) =2^{k-1}+\sigma _{2^{k-2}}^{-1}\left(
m-1\right) \text{ \ with \ }m=1,2,3,\cdots ,2^{k-2}  \label{F_8}
\end{equation}

It is assumed that $\sigma_{2^{k}}^{-1}\left( 0\right) =\sigma_{2^{k}}^{-1}\left( 2^{k}\right) $.

As $\sigma_{2^{k-2}}\left( m\right) \in \left\{ 1,2,3,\ldots,2^{k-2}\right\} $ it results%
\begin{equation*}
\sigma _{2^{k}}^{-1}\left( 4m-4\right) =2^{k-1}+\sigma _{2^{k-2}}^{-1}\left(
m-1\right) \in \left\{ 2^{k-1}+1,2^{k-1}+2,\cdots ,2^{k-1}+2^{k-2}\right\}
\end{equation*}%
or to put it another way, the indexes of the columns are given by%
\begin{equation*}
2^{k-1}+1,2^{k-1}+2,2^{k-1}+3,\cdots ,2^{k-1}+2^{k-2}
\end{equation*}

\textbf{ii) }We want to prove that the submatrix $C_{2^{k-2}}$ of%
$P_{\Omega _{\text{ }}^{-1}\left( \sigma _{2^{k}}^{-1}\right) }$ is formed by
rows $2^{k-2}+1,2^{k-2}+2,2^{k-2}+3,\cdots ,2^{k-1}$ and columns%
$2^{k-1}+2^{k-2}+1,2^{k-1}+2^{k-2}+2,2^{k-1}+2^{k-2}+3,\cdots ,2^{k}.$

According to (\ref{F_5}) and lemma \ref{L_1} it results that
row-indexes with entry $1$ are given by%
\begin{equation}
\begin{array}{cc}
\sigma _{2^{k}}^{-1}\left( 4m-1\right) =2^{k-1}+1-\sigma
_{2^{k-2}}^{-1}\left( m\right) & m=1,2,3,\cdots ,2^{k-2}%
\end{array}
\label{F_27}
\end{equation}%
therefore%
\begin{equation}
\sigma _{2^{k}}^{-1}\left( 4m-1\right) \in \left\{
2^{k-2}+1,2^{k-2}+2,2^{k-2}+3,\cdots ,2^{k-1}\right\}  \label{F_12}
\end{equation}

In the same way, according to (\ref{F_5}) and by using lemma
\ref{L_1}, the column-indexes with entry $1$ are given by%
\begin{equation}
\sigma _{2^{k}}^{-1}\left( 4m-2\right) =2^{k}+1-\sigma _{2^{k-2}}^{-1}\left(
m\right) \text{ \ }m=1,2,3,\cdots ,2^{k-2}  \label{F_28}
\end{equation}%
therefore%
\begin{equation}
\sigma _{2^{k}}^{-1}\left( 4m-2\right) \in \left\{
2^{k-1}+2^{k-2}+1,2^{k-1}+2^{k-2}+2,2^{k-1}+2^{k-2}+3,\cdots ,2^{k}\right\}
\label{F_13}
\end{equation}

It follows that the submatrix $C_{2^{k-2}}$ of
$P_{\Omega^{-1}\left( \sigma _{2^{k}}^{-1}\right) }$
is formed by rows $2^{k-2}+1,2^{k-2}+2,2^{k-2}+3,\cdots ,2^{k-1}$,
and columns $2^{k-1}+2^{k-2}+1,2^{k-1}+2^{k-2}+2,2^{k-1}+2^{k-2}+3,\cdots ,2^{k}$.

\textbf{iii) }We want to prove that $C_{2^{k-1}}^{\ast }$ is the
submatrix of $P_{\Omega^{-1}\left( \sigma _{2^{k}}^{-1}\right) }$
formed by rows $2^{k-1}+1,2^{k-1}+2,2^{k-1}+3,\cdots ,2^{k}$ and columns
$1,2,\cdots ,2^{k-1}$.

According to (\ref{F_6}) and by using (\ref{F_0}) it results
the row-indexes of submatrix with entry $1$ are given by
\begin{equation*}
\sigma _{2^{k}}^{-1}\left( 2m\right) =2^{k}+1-\sigma _{2^{k-1}}^{-1}\left(
m\right) \text{ \ with \ }m=1,2,3,\cdots ,2^{k-1}
\end{equation*}%
as%
\begin{equation*}
\sigma _{2^{k-1}}^{-1}\left( m\right) \in \left\{ 1,2,3,\cdots
,2^{k-1}\right\}
\end{equation*}%
it gives%
\begin{equation}
\sigma _{2^{k}}^{-1}\left( 2m\right) =2^{k}+1-\sigma _{2^{k-1}}\left(
m\right) \in \left\{ 2^{k-1}+1,2^{k-1}+2,2^{k-1}+3,\cdots ,2^{k}\right\}
\label{F_24}
\end{equation}

According to (\ref{F_6}) it results that the column-indexes of submatrix with entry $1$ are given by%
\begin{equation*}
\sigma _{2^{k}}^{-1}\left( 2m-1\right) \text{ , \ }m=1,2,3,\cdots ,2^{k-1}
\end{equation*}

By using (\ref{F_0}) it results%
\begin{equation*}
{\large \sigma }_{2^{k}}^{-1}\left( 2m-1\right) ={\large \sigma }%
_{2^{k-1}}^{-1}\left( m\right) \text{\ , \ }m=1,2,3,\cdots ,2^{k-1}
\end{equation*}
therefore%
\begin{equation}
{\large \sigma }_{2^{k}}^{-1}\left( 2m-1\right) \in \left\{ 1,2,3,\cdots,2^{k-1}\right\}  \label{F_25}
\end{equation}

It yields that the submatrix $C_{2^{k-1}}^{\ast }$ of $P_{\Omega^{-1}\left( \sigma _{2^{k}}^{-1}\right) }$ is formed by rows
$2^{k-1}+1,2^{k-1}+2,2^{k-1}+3,\cdots ,2^{k}$ and columns
$1,2,3,\cdots ,2^{k-1}$
\end{proof}

We are now in position to get the matricial representation in the following theorem.
\begin{theorem}
\label{T_2}Let $A_{2^{k}}$ the matrix referred to in definition
\ref{D_3}. $A_{2^{k}}$ is the permutation matrix associated with the permutation $\Omega^{-1}\left( \sigma _{2^{k}}^{-1}\right)$ 
$\forall k\in
%TCIMACRO{\U{2115} }%
%BeginExpansion
\mathbb{N}
%EndExpansion
$, that is, $P_{\Omega ^{-1}\left( \sigma _{2^{k}}^{-1}\right) }=A_{2^{k}}$. We will say that $A_{2^{k}}$
is the matrix representation of period doubling cascade.
\end{theorem}

\begin{proof}
\linebreak The proof for $k=1$ and $k=2$ is left to the reader.
Let be $k\in
%TCIMACRO{\U{2115} }%
%BeginExpansion
\mathbb{N}
%EndExpansion
$  $k\geq 3$. The theorem will be proven by induction on $k$.

\textbf{i)} Case $k=3$.

According to (\ref{F_15}), and definition \ref{D_3} the
matrix $P_{\Omega ^{-1}\left( \sigma _{2^{3}}^{-1}\right) }$ is%
\begin{equation*}
P_{\Omega ^{-1}\left( \sigma _{2^{3}}^{-1}\right) }=\left(
\begin{tabular}{llll|llll}
$0$ & $0$ & $0$ & $0$ & $0$ & $\mathbf{1}$ & \multicolumn{1}{|l}{$0$} & $0$
\\
$0$ & $0$ & $0$ & $0$ & $\mathbf{1}$ & $0$ & \multicolumn{1}{|l}{$0$} & $0$
\\ \cline{5-8}
$0$ & $0$ & $0$ & $0$ & $0$ & $0$ & \multicolumn{1}{|l}{$\mathbf{1}$} & $0$
\\
$0$ & $0$ & $0$ & $0$ & $0$ & $0$ & \multicolumn{1}{|l}{$0$} & $\mathbf{1}$
\\ \hline
$0$ & $0$ & $0$ & $\mathbf{1}$ & $0$ & $0$ & $0$ & $0$ \\
$0$ & $0$ & $\mathbf{1}$ & $0$ & $0$ & $0$ & $0$ & $0$ \\
$0$ & $\mathbf{1}$ & $0$ & $0$ & $0$ & $0$ & $0$ & $0$ \\
$\mathbf{1}$ & $0$ & $0$ & $0$ & $0$ & $0$ & $0$ & $0$%
\end{tabular}%
\right) =\left(
\begin{tabular}{c|c|c}
$0$ & $A_{2^{3-2}}$ & $0$ \\ \hline
$0$ & $0$ & $I_{2^{3-2}}$ \\ \hline
$I_{2^{3-1}}^{\ast }$ & $0$ & $0$%
\end{tabular}%
\right) =A_{2^{3}}
\end{equation*}%

\textbf{ii)} By hypothesis of induction we assume that $P_{\Omega^{-1}\left( \sigma _{2^{n}}^{-1}\right) }=A_{2^{n}}$,
$3\leq n<k$ and we have to prove that $P_{\Omega ^{-1}\left( \sigma _{2^{k}}^{-1}\right)}=A_{2^{k}}$.

According to lemma \ref{L_3} and definition \ref{D_3} to
prove that $P_{\Omega ^{-1}\left( \sigma _{2^{k}}^{-1}\right) }=A_{2^{k}}$
it is enough to prove that $B_{2^{k-2}}=A_{2^{k-2}}$,
$C_{2^{k-1}}^{\ast }=I_{2^{k-1}}^{\ast }$ and $C_{2^{k-2}}=I_{2^{k-2}}$.

\textbf{a) }We want to prove that the submatrices $B_{2^{k-2}}$
and $A_{2^{k-2}}$ are the same, that is, $B_{2^{k-2}}=A_{2^{k-2}}$.

According to (\ref{F_4}), $B_{2^{k-2}}=\left( b_{ij}\right) $ with
\begin{equation*}
b_{ij}=\left\{
\begin{array}{cl}
1 & \text{ if }i=\sigma _{2^{k}}^{-1}\left( 4m-3\right) \text{ \ and \ }%
j=\sigma _{2^{k}}^{-1}\left( 4m-4\right) \text{, }m=1,2,3,\cdots ,2^{k-2} \\
0 & \text{ otherwise}%
\end{array}%
\right.
\end{equation*}

By using (\ref{F_7}) and (\ref{F_8}) it results that
$\sigma _{2^{k}}^{-1}\left( 4m-3\right) =\sigma _{2^{k-2}}^{-1}\left(
m\right) $ and $\ \sigma _{2^{k}}^{-1}\left( 4m-4\right) =2^{k-1}+\sigma
_{2^{k-2}}^{-1}\left( m-1\right) $ $\ \ m=1,2,3,\cdots ,2^{k-2},$ $k\geq 2$.
Therefore it follows that%
\begin{equation*}
b_{ij}=\left\{
\begin{array}{cl}
1 & \text{ if }i=\sigma _{2^{k-2}}^{-1}\left( m\right) \text{ \ and \ }%
j=2^{k-1}+\sigma _{2^{k-2}}^{-1}\left( m-1\right) \text{, }m=1,2,3,\cdots
,2^{k-2} \\
0 & \text{ otherwise}%
\end{array}%
\right.
\end{equation*}

If we do not take into account the term $2^{k-1}$, which appears in the
column-index $"j"$ and represents a displacement of the columns of $B_{2^{k-2}}$
in the matrix $P_{\Omega _{\text{ }}^{-1}\left( \sigma_{2^{k}}^{-1}\right) }$,
then $B_{2^{k-2}}$ is the associated matrix with
$\Omega ^{-1}\left( \sigma _{2^{k-2}}^{-1}\right)$ (see (\ref{F_15})), that by hypothesis
of induction coincides with $A_{2^{k-2}},$ therefore $B_{2^{k-2}}=A_{2^{k-2}}$.

\textbf{2) }We want to prove that the submatrices $C_{2^{k-2}}$
and $I_{2^{k-2}}$ are the same, that is, $C_{2^{k-2}}=I_{2^{k-2}}$.

From (\ref{F_27}) and (\ref{F_28})
we know that row and column-indexes with entry $1$ in $C_{2^{k-2}}$
are respectively given by
\begin{equation*}
\sigma _{2^{k}}^{-1}\left( 4m-1\right) =2^{k-1}+1-\sigma
_{2^{k-2}}^{-1}\left( m\right) ,\text{ }m=1,2,3,\cdots ,2^{k-2}
\end{equation*}

\begin{equation*}
\sigma _{2^{k}}^{-1}\left( 4m-2\right) =2^{k}+1-\sigma _{2^{k-2}}^{-1}\left(m\right) \text{,\ }m=1,2,3,\cdots ,2^{k-2}
\end{equation*}

Therefore%
\begin{equation*}
\sigma _{2^{k}}^{-1}\left( 4m-1\right) -2^{k-2}=\sigma _{2^{k}}^{-1}\left(4m-2\right) -\left( 2^{k-1}+2^{k-2}\right) ,\text{ }m=1,2,\cdots ,2^{k-2}
\end{equation*}
So taking into account (\ref{F_12}) and (\ref{F_13}) it results
$C_{2^{k-2}}$ is the submatrix with entries $1$ in the diagonal, and displaced $2^{k-2}$ rows
and $2^{k-1}+2^{k-2}$ columns, that is,
$C_{2^{k-2}}=I_{2^{k-2}}$ displaced $2^{k-2}$ rows and $2^{k-1}+2^{k-2}$ columns in the matrix 
$P_{\Omega ^{-1}\left( \sigma _{2^{n}}^{-1}\right) }$.

\textbf{3) }We want to prove that the submatrices 
$C_{2^{k-1}}^{\ast }$ and $I_{2^{k-1}}^{\ast }$ are the same, that is,
$C_{2^{k-1}}^{\ast }=I_{2^{k-1}}^{\ast }$.

According to (\ref{F_6}) we have that
$C_{2^{k-1}}^{\ast }=\left( c_{ij}^{\ast }\right) $ is a $2^{k-1}$
order matrix with
\begin{equation*}
c_{ij}^{\ast }=\left\{
\begin{array}{cl}
1 & \text{ if }i=\sigma _{2^{k}}^{-1}\left( 2m\right) \text{ \ and \ }j=\sigma
_{2^{k}}^{-1}\left( 2m-1\right) \text{, }m=1,2,3,\cdots ,2^{k-1} \\
0 & \text{ otherwise}%
\end{array}%
\right.
\end{equation*}

By using (\ref{F_0}) it follows that%
\begin{equation}
c_{ij}^{\ast }=\left\{
\begin{array}{cl}
1 & \text{ if }i=2^{k}+1-\sigma _{2^{k-1}}^{-1}\left( m\right) \text{ and \ }%
j=\sigma _{2^{k-1}}^{-1}\left( m\right) \text{, }m=1,2,3,\cdots ,2^{k-1} \\
0 & \text{ otherwise}%
\end{array}%
\right.  \label{F_21}
\end{equation}

Let us notice that $i=2^{k-1}+1,2^{k-1}+2,\cdots ,2^{k-1}+2^{k-1}$
and $j=1,2,3,\cdots ,2^{k-1}$ (see (\ref{F_24}) and (\ref{F_25})),
that is, $C_{2^{k-1}}^{\ast }$ is a $2^{k-1}$ order submatrix displaced $2^{k-1}$
rows in the matrix $P_{\Omega ^{-1}\left( \sigma_{2^{n}}^{-1}\right) }$.

If we had not taken into account this displacement (\ref{F_21}) would be rewritten as%
\begin{equation*}
c_{ij}^{\ast }=\left\{
\begin{array}{cl}
1 & \text{ if }i=2^{k-1}+1-\sigma _{2^{k-1}}^{-1}\left( m\right) \text{ and \ }%
j=\sigma _{2^{k-1}}^{-1}\left( m\right) \text{, }m=1,2,3,\cdots ,2^{k-1} \\
0 & \text{ otherwise}%
\end{array}%
\right.
\end{equation*}%
and therefore
\begin{equation*}
i+j=2^{k-1}+1
\end{equation*}%
where $i$ and $j$ are respectively the row and column index
of entry $1$ in the not displaced $C_{2^{k-1}}^{\ast }$
where $i,j=1,2,\cdots ,2^{k-1}$. According to definition \ref{D_3} we have
$I_{2^{k-1}}^{\ast }=\left( a_{ij}\right)$ with
\begin{equation*}
a_{ij}=\left\{
\begin{array}{cl}
1 & \text{ \ if \ }i+j=2^{k-1}+1 \\
0 & \text{ otherwise}%
\end{array}%
i,j=1,2,\cdots ,2^{k-1}\right.
\end{equation*}

It follows that $C_{2^{k-1}}^{\ast }=I_{2^{k-1}}^{\ast }$ and is
displaced $2^{k-1}$ rows in the matrix
$P_{\Omega ^{-1}\left( \sigma _{2^{n}}^{-1}\right) }$.
\end{proof}

An equivalent formulation of theorem \ref{T_2} is given in the
following theorem where no recurrence relations between blocks are used to build the matrix.

\begin{theorem}
Let $I_{2^{k}}$ and $I_{2^{k}}^{\ast }$ the matrices introduced in definition \ref{D_3}.
The matrix representation of the period doubling cascade, $A_{2^{k}},$ $k\in
%TCIMACRO{\U{2115} }%
%BeginExpansion
\mathbb{N}
%EndExpansion
$, is given by

\textbf{i)} $k$ is even:

\begin{equation}
A_{2^{k}}=\left(
\begin{tabular}{l|l}
&
\begin{tabular}{l|l}
\begin{tabular}{l|l}
&
\begin{tabular}{l|l}
\begin{tabular}{l|l}
& $I_{2^{1}}$ \\ \hline
$I_{2^{1}}^{\ast }$ &
\end{tabular}
&  \\ \hline
& $\ddots$ \\ \cline{2-2}
\end{tabular}
\\ \hline
\begin{tabular}{ll}
& $I_{2^{k-3}}^{\ast }$ \\
&
\end{tabular}
&  \\ \cline{1-1}
\end{tabular}
&  \\ \hline
&
\begin{tabular}{lll}
&  &  \\
& $I_{2^{k-2}}$ &  \\
&  &
\end{tabular}%
\end{tabular}
\\ \hline
\multicolumn{1}{c|}{$%
\begin{tabular}{lllllll}
&  &  &  &  &  &  \\
&  &  &  &  &  &  \\
&  &  &  &  &  &  \\
&  &  & $I_{2^{k-1}}^{\ast }$ &  &  &  \\
&  &  &  &  &  &  \\
&  &  &  &  &  &  \\
&  &  &  &  &  &
\end{tabular}%
$} &
\end{tabular}%
\right)  \label{F_16}
\end{equation}

\textbf{ii)} $k$ is odd

\begin{equation*}
A_{2^{1}}=\left(
\begin{tabular}{ll}
$0$ & $1$ \\
$1$ & $0$%
\end{tabular}%
\right)
\end{equation*}
and $k\geq 3$%
\begin{equation}
A_{2^{k}}=\left(
\begin{tabular}{l|l}
&
\begin{tabular}{l|l}
\begin{tabular}{l|l}
&
\begin{tabular}{l|l}
\begin{tabular}{l|l}
$I_{2^{1}}^{\ast }$ &  \\ \hline
& $I_{2^{1}}$%
\end{tabular}
&  \\ \hline
& $\mathbf{\ddots }$ \\ \cline{2-2}
\end{tabular}
\\ \hline
\begin{tabular}{ll}
& $I_{2^{k-3}}^{\ast }$ \\
&
\end{tabular}
&  \\ \cline{1-1}
\end{tabular}
&  \\ \hline
&
\begin{tabular}{lll}
&  &  \\
& $I_{2^{k-2}}$ &  \\
&  &
\end{tabular}%
\end{tabular}
\\ \hline
\multicolumn{1}{c|}{$%
\begin{tabular}{lllllll}
&  &  &  &  &  &  \\
&  &  &  &  &  &  \\
&  &  &  &  &  &  \\
&  &  & $I_{2^{k-1}}^{\ast }$ &  &  &  \\
&  &  &  &  &  &  \\
&  &  &  &  &  &  \\
&  &  &  &  &  &
\end{tabular}%
$} &
\end{tabular}%
\right)  \label{F_17}
\end{equation}
\end{theorem}

\begin{proof}

\textbf{i)} Let $k$ be even, so we write $k=2m$, $m\in
%TCIMACRO{\U{2115} }%
%BeginExpansion
\mathbb{N}
%EndExpansion
$ and the theorem is proven by induction on $m$.

\textbf{a)} For $m=1,$ we have
\begin{equation*}
A_{2^{2m}}=A_{2^{2}}=\left(
\begin{tabular}{ll|ll}
&  & $1$ &  \\
&  &  & $1$ \\ \hline
& $1$ &  &  \\
$1$ &  &  &
\end{tabular}%
\right) =\left(
\begin{tabular}{l|l}
& $I_{2^{1}}$ \\ \hline
$I_{2^{1}}^{\ast }$ &
\end{tabular}%
\right)
\end{equation*}

\textbf{b)} By hypothesis of induction we assume that $A_{2^{2n}}$ given by
(\ref{F_16}) is true for $1\leqslant n<m$ and we have to prove
that is true for $n=m$.

On one side, by theorem \ref{T_2} we have
\begin{equation}
A_{2^{2m}}=\left(
\begin{tabular}{c|c|c}
$0$ & $A_{2^{2m-2}}$ & $0$ \\ \hline
$0$ & $0$ & $I_{2^{2m-2}}$ \\ \hline
$I_{2^{2m-1}}^{\ast }$ & $0$ & $0$%
\end{tabular}%
\right) =\left(
\begin{tabular}{c|c|c}
$0$ & $A_{2^{2\left( m-1\right) }}$ & $0$ \\ \hline
$0$ & $0$ & $I_{2^{2m-2}}$ \\ \hline
$I_{2^{2m-1}}^{\ast }$ & $0$ & $0$%
\end{tabular}%
\right)   \label{F_22}
\end{equation}%
On the other side, as $2\left( m-1\right) $ is even and $\left(m-1\right) <m$,
by using the hypothesis of induction it results
\begin{equation}
A_{2^{2\left( m-1\right) }}=\left(
\begin{tabular}{c|c|c}
& $%
\begin{tabular}{l|l}
& $I_{2^{1}}$ \\ \hline
$I_{2^{1}}^{\ast }$ &
\end{tabular}%
$ &  \\ \hline
&  & $\mathbf{\ddots }$ \\ \hline
\begin{tabular}{ll}
&  \\
$I_{2^{2m-3}}^{\ast }$ &
\end{tabular}
&  &
\end{tabular}%
\right)   \label{F_23}
\end{equation}%
Then, after substituting $\left( \ref{F_23}\right)$ in $\left( \ref{F_22}\right) $ it yields

$A_{2^{2m}}=\left(
\begin{tabular}{l|l}
&
\begin{tabular}{l|l}
\begin{tabular}{l|l}
&
\begin{tabular}{l|l}
\begin{tabular}{l|l}
& $I_{2^{1}}$ \\ \hline
$I_{2^{1}}^{\ast }$ &
\end{tabular}
&  \\ \hline
& $\mathbf{\ddots }$ \\ \cline{2-2}
\end{tabular}
\\ \hline
\begin{tabular}{ll}
& $I_{2^{2m-3}}^{\ast }$ \\
&
\end{tabular}
&  \\ \cline{1-1}
\end{tabular}
&  \\ \hline
&
\begin{tabular}{ll}
&  \\
& $I_{2^{2m-2}}$ \\
&
\end{tabular}%
\end{tabular}
\\ \hline
\multicolumn{1}{c|}{$%
\begin{tabular}{lllllll}
&  &  &  &  &  &  \\
&  &  &  &  &  &  \\
&  &  &  &  &  &  \\
&  &  & $I_{2^{2m-1}}^{\ast }$ &  &  &  \\
&  &  &  &  &  &  \\
&  &  &  &  &  &  \\
&  &  &  &  &  &
\end{tabular}%
$} &
\end{tabular}%
\right) $

and as $k=2m$ it finally results
\begin{equation*}
A_{2^{k}}=\left(
\begin{tabular}{l|l}
&
\begin{tabular}{l|l}
\begin{tabular}{l|l}
&
\begin{tabular}{l|l}
\begin{tabular}{l|l}
& $I_{2^{1}}$ \\ \hline
$I_{2^{1}}^{\ast }$ &
\end{tabular}
&  \\ \hline
& $\mathbf{\ddots }$ \\ \cline{2-2}
\end{tabular}
\\ \hline
\begin{tabular}{ll}
& $I_{2^{k-3}}^{\ast }$ \\
&
\end{tabular}
&  \\ \cline{1-1}
\end{tabular}
&  \\ \hline
&
\begin{tabular}{llll}
&  &  &  \\
&  &  &  \\
&  & $I_{2^{k-2}}$ &  \\
&  &  &
\end{tabular}%
\end{tabular}
\\ \hline
\multicolumn{1}{c|}{$%
\begin{tabular}{lllllll}
&  &  &  &  &  &  \\
&  &  &  &  &  &  \\
&  &  &  &  &  &  \\
&  &  & $I_{2^{k-1}}^{\ast }$ &  &  &  \\
&  &  &  &  &  &  \\
&  &  &  &  &  &  \\
&  &  &  &  &  &
\end{tabular}%
$} &
\end{tabular}%
\right)
\end{equation*}

\textbf{ii)} Let $k$ be odd, so we write $k=2m-1$, $m\in
%TCIMACRO{\U{2115} }%
%BeginExpansion
\mathbb{N}
%EndExpansion
$, and the proof is proven by induction on $m$.

\textbf{a)} For $m=2$ we have
\begin{equation*}
A_{2^{3}}=\left(
\begin{tabular}{llll|llll}
$0$ & $0$ & $0$ & $0$ & $0$ & $\mathbf{1}$ & \multicolumn{1}{|l}{$0$} & $0$
\\
$0$ & $0$ & $0$ & $0$ & $\mathbf{1}$ & $0$ & \multicolumn{1}{|l}{$0$} & $0$
\\ \cline{5-8}
$0$ & $0$ & $0$ & $0$ & $0$ & $0$ & \multicolumn{1}{|l}{$\mathbf{1}$} & $0$
\\
$0$ & $0$ & $0$ & $0$ & $0$ & $0$ & \multicolumn{1}{|l}{$0$} & $\mathbf{1}$
\\ \hline
$0$ & $0$ & $0$ & $\mathbf{1}$ & $0$ & $0$ & $0$ & $0$ \\
$0$ & $0$ & $\mathbf{1}$ & $0$ & $0$ & $0$ & $0$ & $0$ \\
$0$ & $\mathbf{1}$ & $0$ & $0$ & $0$ & $0$ & $0$ & $0$ \\
$\mathbf{1}$ & $0$ & $0$ & $0$ & $0$ & $0$ & $0$ & $0$%
\end{tabular}%
\right) =\left(
\begin{tabular}{llll|llll}
&  &  &  & $I_{2^{3-2}}^{\ast }$ &  & \multicolumn{1}{|l}{} &  \\
&  &  &  &  &  & \multicolumn{1}{|l}{} &  \\ \cline{5-8}
&  &  &  &  &  & \multicolumn{1}{|l}{} &  \\
&  &  &  &  &  & \multicolumn{1}{|l}{} & $I_{2^{3-2}}$ \\ \hline
&  &  &  &  &  &  &  \\
&  & $I_{2^{3-1}}^{\ast }$ &  &  &  &  &  \\
&  &  &  &  &  &  &  \\
&  &  &  &  &  &  &
\end{tabular}%
\right)
\end{equation*}

\textbf{b)} By hypothesis of induction we assume that $A_{2^{2n-1}}$ given by
(\ref{F_17}) is true for $2 \leq n<m$ and we have to prove
that it is true for $n=m$.

On one side, by theorem \ref{T_2} we have $A_{2^{k}}=\left(
\begin{tabular}{c|c|c}
$0$ & $A_{2^{k-2}}$ & $0$ \\ \hline
$0$ & $0$ & $I_{2^{k-2}}$ \\ \hline
$I_{2^{k-1}}^{\ast }$ & $0$ & $0$%
\end{tabular}%
\right) $ therefore
\begin{equation}
A_{2^{2m-1}}=\left(
\begin{tabular}{c|c|c}
$0$ & $A_{2^{\left( 2m-1\right) -2}}$ & $0$ \\ \hline
$0$ & $0$ & $I_{2^{\left( 2m-1\right) -2}}$ \\ \hline
$I_{2^{\left( 2m-1\right) -1}}^{\ast }$ & $0$ & $0$%
\end{tabular}%
\right) =\left(
\begin{tabular}{c|c|c}
$0$ & $A_{2^{2\left( m-1\right) -1}}$ & $0$ \\ \hline
$0$ & $0$ & $I_{2^{\left( 2m-1\right) -2}}$ \\ \hline
$I_{2^{\left( 2m-1\right) -1}}^{\ast }$ & $0$ & $0$%
\end{tabular}%
\right)   \label{F_29}
\end{equation}

On the other side as $\left( 2m-1\right) -2=2\left( m-1\right) -1$
is odd and $m-1<m$, by using the induction hypothesis it results
\begin{equation}
A_{2^{2\left( m-1\right) -1}}=\left(
\begin{tabular}{c|c|c}
&
\begin{tabular}{l|l}
$I_{2^{1}}$ &  \\ \hline
& $I_{2^{1}}^{\ast }$%
\end{tabular}
&  \\ \hline
&  & $\mathbf{\ddots }$ \\ \hline
\begin{tabular}{ll}
&  \\
$I_{2^{^{2\left( m-1\right) -1-1}}}^{\ast }$ &
\end{tabular}
&  &
\end{tabular}%
\right)   \label{F_30}
\end{equation}%
then after substituting (\ref{F_30}) in (\ref{F_29}) we finally get
\begin{equation*}
A_{2^{k}}=\left(
\begin{tabular}{l|l}
&
\begin{tabular}{l|l}
\begin{tabular}{l|l}
&
\begin{tabular}{l|l}
\begin{tabular}{l|l}
$I_{2^{1}}^{\ast }$ &  \\ \hline
& $I_{2^{1}}$%
\end{tabular}
&  \\ \hline
& $\mathbf{\ddots }$ \\ \cline{2-2}
\end{tabular}
\\ \hline
\begin{tabular}{ll}
& $I_{2^{k-3}}^{\ast }$ \\
&
\end{tabular}
&  \\ \cline{1-1}
\end{tabular}
&  \\ \hline
&
\begin{tabular}{llll}
&  &  &  \\
&  &  &  \\
&  & $I_{2^{k-2}}$ &  \\
&  &  &
\end{tabular}%
\end{tabular}
\\ \hline
\multicolumn{1}{c|}{$%
\begin{tabular}{lllllll}
&  &  &  &  &  &  \\
&  &  &  &  &  &  \\
&  &  &  &  &  &  \\
&  &  & $I_{2^{k-1}}^{\ast }$ &  &  &  \\
&  &  &  &  &  &  \\
&  &  &  &  &  &  \\
&  &  &  &  &  &
\end{tabular}%
$} &
\end{tabular}%
\right)
\end{equation*}
for $k=2m-1$, as we wanted to prove.
\end{proof}

\section{Discussion}

The matricial representation leads to some practical as well as fundamental considerations.
Researches in any field have a good background in linear algebra and it is important to build bridges with other disciplines.
However, the main advantage of matricial representation on permutation
$\sigma _{2^{k}}^{-1}$ is that it can handle a set of particles, instead of
one as $\sigma _{2^{k}}^{-1}$ does; furthermore, the set of particles can
have physical properties. Let be a physical system such that its states
undergo a period doubling cascade. For the sake of simplicity we suppose
that there have been two doubling bifurcations, that is $k=2$, so
\begin{equation*}
A_{2^{2}}=%
\begin{pmatrix}
0 & 0 & 1 & 0 \\
0 & 0 & 0 & 1 \\
0 & 1 & 0 & 0 \\
1 & 0 & 0 & 0%
\end{pmatrix}%
\text{ and\textbf{\ }}\sigma _{2^{2}}^{-1}=%
\begin{pmatrix}
1 & 2 & 3 & 4 \\
1 & 4 & 2 & 3%
\end{pmatrix}%
\end{equation*}%
If the physical system has three particles with properties $a,b,c$ in the states $1,2$ and $4$ we write
$\begin{pmatrix}
a \\
b \\
0 \\
c%
\end{pmatrix}$ and its evolution is given by
$\begin{pmatrix}
0 & 0 & 1 & 0 \\
0 & 0 & 0 & 1 \\
0 & 1 & 0 & 0 \\
1 & 0 & 0 & 0%
\end{pmatrix}%
\begin{pmatrix}
a \\
b \\
0 \\
c%
\end{pmatrix}%
=%
\begin{pmatrix}
0 \\
c \\
b \\
a%
\end{pmatrix}$, that is, the particle with the property $"a"$ has moved from state $1$ to $4$ and so on.
If we wanted to calculate the evolution by using $\sigma _{2^{2}}^{-1}$
we would have to use the permutation three times, furthermore the permutation cannot manage properties in a particle.
Obviously, as particles and states increase the use of $\sigma _{2^{k}}^{-1}$ is completely useless
from a practical point of view.

On the other hand we could write
$\begin{pmatrix}
\frac{1}{2} \\
\frac{1}{2} \\
-\frac{1}{2} \\
-\frac{1}{2}%
\end{pmatrix}$
to indicate that there are two particles with spin $\frac{1}{2}$
in the states $1$ and $2$, and two particles with spin $-\frac{1}{2}$ in
the states $3$ and $4$. So the evolution is%
\begin{equation}
\mathbf{\ }%
\begin{pmatrix}
0 & 0 & 1 & 0 \\
0 & 0 & 0 & 1 \\
0 & 1 & 0 & 0 \\
1 & 0 & 0 & 0%
\end{pmatrix}%
\begin{pmatrix}
\frac{1}{2} \\
\frac{1}{2} \\
-\frac{1}{2} \\
-\frac{1}{2}%
\end{pmatrix}%
=-%
\begin{pmatrix}
\frac{1}{2} \\
\frac{1}{2} \\
-\frac{1}{2} \\
-\frac{1}{2}%
\end{pmatrix}
\label{F_31}
\end{equation}

However $\sigma _{2^{k}}^{-1}$ can not handle the evolution of particles
with spin or any other physical property. If we pay attention to (\ref{F_31})
we see an important fact: eigenvectors and eigenvalues. This
leads us to one fundamental question:
What is the physical meaning of eigenvalues and eigenvectors of $A_{2^{k}}$?
On the other hand any state of the system can be decomposed into a linear combination of eigenvectors.

The idea $k\rightarrow \infty $ has been discussed in the introduction, in order
to understand analytically the transition to chaos. This same idea makes one
wonder what is the relationship between $A_{2^{k}},k\rightarrow \infty $
and the matricial representation of the linearized period-doubling operator,
both of them underlying in period doubling cascade.

\end{document}